\documentclass[11pt,a4paper]{article}
\usepackage[left=15mm, right=15mm, top=20mm,bottom=20mm]{geometry}
\usepackage{amsmath}
\usepackage{amsthm}
\usepackage{ mathrsfs }
\usepackage{amssymb}
\usepackage{xcolor}

\newtheorem{theorem}{Theorem}

\newtheorem{lemma}[theorem]{Lemma}

\newtheorem{remark}[theorem]{Remark}

\numberwithin{equation}{section}
\numberwithin{theorem}{section}

\DeclareMathOperator{\Div}{Div}
\DeclareMathOperator{\pr}{pr}
\DeclareMathOperator{\E}{E}
\DeclareMathOperator{\D}{D}

\begin{document}

\title{Variational symmetries and Lagrangian multiforms}

\author{Duncan Sleigh, Frank Nijhoff and Vincent Caudrelier\\ 
School of Mathematics, University of Leeds}

\maketitle

\abstract{By considering the closure property of a Lagrangian multiform as a conservation law, we use Noether's theorem to show that every variational symmetry of a Lagrangian leads to a Lagrangian multiform.  In doing so, we provide a systematic method for constructing Lagrangian multiforms for which the closure property and the multiform Euler-Lagrange (EL) both hold.  We present three examples, including the first known example of a Lagrangian 3-form: a multiform for the Kadomtsev{-}Petviashvili equation.  We also present a new proof of the multiform EL equations for a Lagrangian k-form for arbitrary k.}

\section{Introduction}

When considering integrable systems, a key weakness of the conventional Lagrangian description is that it does not capture multidimensional consistency - the fact that the equations of motion can be seen as members of a hierarchy of compatible equations  which can be simultaneously imposed on the same dependent variables.  A classical Lagrangian functional will only provide one single equation of the motion per component of the system, with no clear connection to the other equations of the hierarchy.  This weakness was overcome in the paper \cite{Lobb2009} where it was proposed to extend the scalar Lagrangian 
\begin{equation}
\mathscr{L}(x,u^{(n)})\textsf{d}x_1\wedge\ldots\wedge \textsf{d}x_k,
\end{equation}
a volume form on a $k$-dimensional base manifold, to a differential $k$-form
\begin{equation}\label{form}
\textsf{L}=\sum_{1\leq i_1<\ldots <i_k\leq N}\mathscr{L}_{(i_1\ldots i_k)}(x,u^{(n)})\ \textsf{d}x_{i_1}\wedge\ldots\wedge \textsf{d}x_{i_k}.
\end{equation}
on a $N$ dimensional base manifold with $k < N$\footnote{Note that in principle we are often working in an arbitrary number of dimensions,
determined by the number of flows of a given integrable hierarchy that we include in our multiform.}.  We use the notation $u^{(n)}$ to represent $u$ and its derivatives up to the $n^{th}$ order. This led to the introduction of a new notion of a Lagrangian multiform, where the multidimensional consistency manifests itself by the action
\begin{equation}
S[u;\sigma]=\int_{\sigma}\textsf{L}(x,u^{(n)})
\end{equation}
having a critical point $u$, such that $u$ is simultaneously a critical point for every choice of the surface of integration $\sigma$, and also that the action $S$ is invariant with respect to interior deformations of the surface of integration.  The first of these conditions is equivalent to the requirement that $\delta \textsf{dL}=0$ and defines the equations of motion known as the multiform Euler-Lagrange equations\footnote{See \eqref{deltadL} for an explanation of this notation.}.  The second of these conditions gives us the closure relation that, {\it on the equations of motion}, $\textsf{dL}=0$ (this follows from Stokes' theorem).  We shall call a differential form $\textsf{L}$ of the type given in \eqref{form} a \textbf{Lagrangian multiform} if $\textsf{dL}=0$ on the equations defined by $\delta \textsf{dL}=0$.  If the solution $u$ defined by $\delta \textsf{dL}=0$ is the zero function, or $\textsf{dL}=0$ for any $u$ we consider our multiform to be trivial.\\

\noindent The full form of the multiform Euler-Lagrange equations for a Lagrangian $k$-form is given in Appendix \ref{MFeqns}.  These equations require that the usual EL equations hold for each coefficient $\mathscr{L}_{(i\ldots j)}$ of the multiform as well as additional relations between the different coefficients.

\begin{remark}
We shall often use the notation $\mathscr{L}_{(i\ldots j)}$ to represent the coefficient of $\textsf{d}x_i\wedge \ldots\wedge \textsf{d}x_j$ in a Lagrangian multiform \textrm{L} (e.g. $\mathscr{L}_{(123)}$ would be the coefficient of $\textsf{d}x_1\wedge \textsf{d}x_2\wedge \textsf{d}x_3$).  We need only define the 
$\mathscr{L}_{(i\ldots j)}$ in the case where $i<\ldots<j$.  We then define the $\mathscr{L}_{(i\ldots j)}$ for other permutations of indices by the convention that they are anti-symmetric.  There are examples of Lagrangian multiforms, such as those given in \cite{Lobb2009}, \cite {Xenitidis2011} and \cite{Sleigh2019}, where there is a natural covariance and anti-symmetry built into the structure such that it is automatic that $\mathscr{L}_{(ij)}=-\mathscr{L}_{(ji)}$.  In the case where we are considering an $N-1$ form on an $N$ dimensional base manifold, we shall also use the notation $\mathscr{L}_{(\bar i)}$ to represent the coefficient of $\textsf{d}x_{i+1}\wedge \ldots\wedge \textsf{d}x_N\wedge \textsf{d}x_1\wedge \ldots\wedge \textsf{d}x_{i-1}$, i.e. where the $\textsf{d}x_{j}$'s appear in cyclic order and $\textsf{d}{x_i}$ is removed.
\end{remark}

\noindent A major difficulty in studying Lagrangian multiforms (particularly when working with Lagrangians that are not naturally covariant) is the construction of the components $\mathscr{L}_{(i\ldots j)}$, even for known integrable classical field theories. This problem has attracted attention previously, e.g. in \cite{Vermeeren2019a}. In this paper, we introduce a new method to answer this problem based on the use of variational symmetries and Noether's theorem \cite{Noether}. 
We note that the connection between Noether's theorem and Lagrangian multiforms was first explored in \cite{Suris2016a}, and extended in \cite{Suris2017} where a systematic method of constructing Lagrangian 1-forms from variational symmetries was given for systems in classical mechanics.  In this paper, we deal with field theories in $1+1$ and, for the first time $2+1$ dimensions. Because we require that $\textsf{dL}=0$ on the equations of motion, we are able to consider this as a conservation law and use Noether's theorem \cite{Noether} to relate this to variational symmetries of the components $\mathscr{L}_{(i\ldots j)}$ of our multiform.
This provides us with a systematic means of constructing Lagrangian multiforms (of any order).  
In Section \ref{VarsymandN} we give a brief overview of variational symmetries, and Noether's theorem.  In Section \ref{mainsect}, we present our new results along with three examples, including a multiform for the first two flows of the K{-}P hierarchy - the first ever example of a continuous $2+1$ dimensional Lagrangian multiform.  In Appendix \ref{MFeqns}, we provide a new proof of the multiform Euler-Lagrange equations for a Lagrangian $k$-form, which were first derived in \cite{Vermeeren2018}.

\section{Variational symmetries and Noether's theorem}\label{VarsymandN}

In this section, we shall make use of a version of Noether's (first) theorem as presented in \cite{Olver1993}, where proofs of all statements in this section can be found.  We consider systems with $p$ independent variables $x=(x_1,\ldots, x_p)$ and $q$ dependent variables $u=(u^1,\ldots ,u^q)^T$. In the rest of this paper, we will often use $u$ to denote the collection of fields $u^1,\ldots ,u^q$ or the vector $(u^1,\ldots ,u^q)^T$.

\subsection{Generalized and evolutionary vector fields}

We consider vector fields of the form
\begin{equation}\label{vector}
\textbf{v}=\sum_{i=1}^p\xi_i\frac{\partial}{\partial x_i}+\sum_{\alpha=1}^q\phi_\alpha\frac{\partial}{\partial u^\alpha}
\end{equation}
We say that \textbf{v} is a \textbf{geometric vector field} if the $\xi_i$ and $\phi_\alpha$ depend only on $x$ and $u$.  If the $\xi_i$ and $\phi_\alpha$ depend also on derivatives of $u$, we say that \textbf{v} is a \textbf{generalized vector field}.  If all of the $\xi_i$ are zero, i.e.

\begin{equation}
\textbf{v}_Q=\sum_{\alpha=1}^q Q_\alpha\frac{\partial}{\partial u^\alpha}\equiv Q\cdot\frac{\partial}{\partial u}\,,
\end{equation}
we call $\textbf{v}_Q$ an \textbf{evolutionary vector field} with \textbf{characteristic} $Q(x,u^{(n)})=(Q_1(x,u^{(n)}),\ldots ,Q_q(x,u^{(n)}))^T$, where $Q(x,u^{(n)})$ is taken to mean that $Q$ may depend on $x$, $u$ and derivatives of $u$.  The prolongation of an evolutionary vector field $\textbf{v}_Q$ takes the form

\begin{equation}
\pr \textbf{v}_Q=\sum_{\alpha,J}\D_JQ_\alpha\frac{\partial}{\partial u_J^\alpha}
\end{equation}
where we have used the multi-index notation where $J$ is the ordered set $(j_1,\ldots,j_p)$ and 
\begin{equation}
\D_J:=\prod _{i=1}^p(\D_{x_i})^{j_i}\,,~~\D_{x_i}=\frac{\partial}{\partial x_i}+\sum_{\alpha,J}u_{Ji}^\alpha \frac{\partial}{\partial u_J^\alpha}\,.
\end{equation}
We shall write $Ji^r$ to denote $(j_1,\ldots,j_i+r,\ldots,j_p)$, $J\backslash k^r$ to denote $(j_1,\ldots,j_k-r,\ldots,j_p)$ and $|J|$ to denote the sum $ j_1+\ldots +j_p$.\\

\noindent Every vector field \textbf{v} in the form of \eqref{vector} has an associated evolutionary representative $\textbf{v}_Q$ where

\begin{equation}\label{evolrep}
Q_\alpha=\phi_\alpha-\sum_{i=1}^p\xi_iu^\alpha_{x_i}
\end{equation}

\subsection{Variational symmetries}

The vector field $\textbf{v}$ is a variational symmetry of a Lagrangian $\mathscr{L}(x,u^{(n)})\textsf{d}x_i\wedge \ldots \wedge \textsf{d}x_j$ if and only if 
	
\begin{equation}
\pr \textbf{v}(\mathscr{L})+\mathscr{L}\Div \xi=\Div B
\end{equation}
for some $B(x,u^{(n)})=(B_1(x,u^{(n)}),\dots,B_p(x,u^{(n)}))^T$.  For an evolutionary vector $\textbf{v}_Q$, this simplifies to
\begin{equation}
\pr \textbf{v}_Q(\mathscr{L})=\Div \tilde B
\end{equation}
for some $\tilde B(x,u^{(n)})=(\tilde B_1(x,u^{(n)}),\dots,\tilde B_p(x,u^{(n)}))^T$.  A generalized vector field \textbf{v} is a variational symmetry of $\mathscr{L}$ if and only if its evolutionary representative $\textbf{v}_Q$ is.\\

\noindent Finding the variational symmetries of a given Lagrangian is a non-trivial exercise.  Methods for doing so are covered in \cite{Olver1993}, \cite{Stephani1990}, \cite{Hydon2000} and \cite{Bluman2002}. In our approach, we assume that such a variational symmetry is given (by applying one of those methods for instance) and we use it as our starting point to construct a Lagrangian multiform.

\subsection{Noether's theorem}

In order to introduce Noether's theorem, we will require the Euler operater $\E$.  We define the Euler operator $\E$ to be the $q$-component vector operator whose $\alpha^{th}$ component is $\E_\alpha$ given by

\begin{equation}
 \E_\alpha=\sum_J(-1)^{|J|}\D_J\frac{\partial}{\partial u^\alpha_J}
 \end{equation}
The sum is over all multi-indices $J=(j_1,\ldots,j_p)$.  For a Lagrangian $\mathscr{L}$, $\E(\mathscr{L})=0$ gives the standard Euler Lagrange equations for $\mathscr{L}$.  For example, in the case where $p=2$, $q=1$ and $\mathscr{L}$ contains terms up to the $2^{nd}$ jet,

\begin{equation}
\E(\mathscr{L})=\frac{\partial \mathscr{L}}{\partial u}-\D_{x_1}\frac{\partial \mathscr{L}}{\partial u_{x_1}}-\D_{x_2}\frac{\partial \mathscr{L}}{\partial u_{x_2}}+\D_{x_1}^2\frac{\partial \mathscr{L}}{\partial u_{x_1x_1}}+\D_{x_1}\D_{x_2}\frac{\partial \mathscr{L}}{\partial u_{x_1x_2}}+\D_{x_2}^2\frac{\partial \mathscr{L}}{\partial u_{x_2x_2}}.
\end{equation}
We say that the equations of motion given by $\E(\mathscr{L})=0$ are of maximal rank if the $q \times (p+q{p+n\choose n})$ Jacobian matrix

\begin{equation}
\textsf{J}_{\E(\mathscr{L})}=\bigg(\frac{\partial \E_i(\mathscr{L})}{\partial x_j},\frac{\partial \E_i(\mathscr{L})}{\partial u_J^\alpha}\bigg)
\end{equation}
is of rank $q$ (i.e. of maximal rank) on the equations of motion given by $\E(\mathscr{L})=0$.

\begin{theorem}{\bf [Noether]} Let $\textbf{v}_Q$ be an evolutionary vector field with characteristic $Q$ and $\mathscr{L}$ a Lagrangian density, such that $\E(\mathscr{L})$ is of maximal rank. Then,
\begin{equation}\label{Noether}
\pr \textbf{v}_Q(\mathscr{L})=\Div B(x,u^{(n)})~~\text{for some}~B\quad \iff\quad Q\cdot \E(\mathscr{L})=\Div P~~\text{for some}~P(x,u^{(n)})\,.
\end{equation}
where $\displaystyle Q\cdot \E=\sum_{\alpha=1}^{q}Q_\alpha \E_\alpha$.
\end{theorem} 

\noindent The right hand side of \eqref{Noether} is the characteristic form of a conservation law.  Since setting $\E(\mathscr{L})=0$ defines the equations of motion, this tells us that $\Div P=0$ on the equations of motion - the usual form of a conservation law.

\section{Variational symmetries as Lagrangian multiforms}\label{mainsect}

In this section, we shall take the well known results of the previous section, and apply them in the context of Lagrangian multiforms.  We consider the Lagrangian density $\mathscr{L}$ on a manifold with $p$ independent, and $q$ dependent variables from the previous section.  In order to be able to apply Noether's theorem, we require that the corresponding EL equations $\E(\mathscr{L})=0$ are of maximal rank.  If we introduce a new independent variable  $x_{p+1}$, independent of $x_1,\ldots,x_p$, and the vector field $\textbf{w}=u_{x_{p+1}}\cdot\dfrac{\partial}{\partial u}$ then

\begin{equation}\label{xb0}
\pr\textbf{w} (\mathscr{L})=\D_{x_{p+1}} \mathscr{L}.
\end{equation}
Also, by reversing the integration by parts that was used to get from $\mathscr{L}$ to $\E(\mathscr{L})$ it follows that

\begin{equation}\label{xb1}
u_{x_{p+1}}\cdot \E(\mathscr{L})=\D_{x_{p+1}} \mathscr{L}+\Div A
\end{equation}
for some $A$, where the $x_{p+1}$ component of $A$ is zero.  If $Q$ is the characteristic of a variational symmetry of $\mathscr{L}$ then Noether's theorem tells us that 

\begin{equation}\label{xb2}
Q\cdot \E(\mathscr{L})=\Div P
\end{equation}
for some $P$.  Adding \eqref{xb1} and \eqref{xb2} gives us that

\begin{equation}\label{xb3}
(u_{x_{p+1}}+Q)\cdot \E(\mathscr{L})=\Div \tilde P
\end{equation}
where $\tilde P=A+P$ so the $x_{p+1}$ component of $\tilde P$ is $\mathscr{L}$.  We use this idea to construct Lagrangian multiforms as follows.

\begin{theorem}\label{mainthm}
Let $Q(x,u^{(n)})$ be the characteristic of a variational symmetry of the Lagrangian density $\mathscr{L}(x,u^{(n)})$ such that $\mathscr{L}$ and $Q$ have no dependence on $x_{p+1}$ or derivatives of $u$ with respect to $x_{p+1}$.  If $\tilde Q=u_{x_{p+1}}+Q$ then 
\begin{equation}\label{MFchar}
\tilde Q\cdot \E(\mathscr{L})=\Div P
\end{equation}
for some $P=(P_1,\ldots P_p,P_{p+1})^T$, and the $p$-form $\textsf{L}$ such that
\begin{equation}
\textsf{L}=\sum_{i=1}^{p+1}\mathscr{L}_{(\bar i)}\mathsf{d}x_{i+1}\wedge\ldots\wedge\mathsf{d}x_{p+1}\wedge\mathsf{d}x_1\wedge\ldots\wedge\mathsf{d}x_{i-1}~~\text{with} ~~\mathscr{L}_{(\bar i)}=(-1)^{ip}P_i
\end{equation}
 is a Lagrangian multiform.  The $p+1$ component of $P$ is equivalent (i.e. equal modulo total derivatives) to $\mathscr{L}$. 
\end{theorem}

\begin{proof}
The existence of a $P$ that satisfies \eqref{MFchar} and has $\mathscr{L}$ as its ${p+1}$ component follows from the introduction to this section, equations \eqref{xb0} to \eqref{xb3}.  Since $Q$ is a symmetry of $\E(\mathscr{L})$ we know that the equations $\tilde Q=0$ and $\E(\mathscr{L})=0$ are compatible in the sense that there exists a \textit{general} common solution.  Then
\begin{equation}
\textsf{dL}=(-1)^{p}\Div P \ \mathsf{d}x_1\wedge\ldots\wedge\mathsf{d}x_{p+1},
\end{equation}
and it follows that $\delta \textrm{dL}=0$ is equivalent to the requirement that
\begin{equation}
\frac{\partial}{\partial u_I}\Div P=0 \quad\forall I.
\end{equation}
Using \eqref{MFchar}, this gives us that
\begin{equation}
\frac{\partial}{\partial u_I}\Div P=\bigg(\frac{\partial}{\partial u_I}\tilde Q\bigg)\cdot \E(\mathscr{L})+\tilde Q\cdot\bigg(\frac{\partial}{\partial u_I} \E(\mathscr{L})\bigg)\,,
\end{equation}
and since $\E(\mathscr{L})$ is of maximal rank (a requirement for Noether's theorem), the necessary and sufficient condition for $\delta \textrm{dL}=0$ is that both $\tilde Q=0$ and $\E(\mathscr{L})=0$ hold simultaneously.  From the form of \eqref{MFchar}, it is clear that $\textsf{dL}=0$ on solutions of either $\tilde Q=0$ or $\E(\mathscr{L})=0$.
\end{proof}

\begin{remark}\label{extMF}
Theorem \ref{mainthm} allows us to construct a $p+1$ dimensional Lagrangian multiform from a Lagrangian in $p$ dimensions and a single variational symmetry.  It is natural to consider whether, in the case where we have a set of $l$ commuting variational symmetries, we can iterate the process to find a $p+l$ dimensional Lagrangian multiform, as was achieved for a class of 1-forms in \cite{Suris2017}.  In Section \ref{AKNSex} we use Theorem \ref{mainthm} to obtain a multiform that incorporates the first three flows of the AKNS hierarchy.  We also show why, in the case of a Lagrangian 2-form, it is always possible to obtain a $2+ l$ dimensional Lagrangian 2-form from an autonomous polynomial Lagrangian $\mathscr{L}_{(12)}$ and a set of $l$ commuting variational symmetries with autonomous polynomial characteristics.  A similar argument can be used for autonomous polynomial $k$-forms for arbitrary $k$.  Whether or not non-autonomous, non-polynomial systems can be extended through repeated application of Theorem \ref{mainthm} remains an open problem.
\end{remark}

\noindent We note that $P$ is not unique.  Indeed, any change to $P$ that is equivalent to adding an exact form to \textrm{L} will also satisfy \eqref{MFchar}.  In addition, we can perform ``integration by parts'' on the left hand side of \eqref{MFchar} and the remaining terms will still be a divergence, e.g.

\begin{equation}\label{trans1}
\tilde Q\cdot \E(\mathscr{L}) \to -\D_x\tilde Q\cdot \D_x^{-1}\E(\mathscr{L}) \text{ and }\Div P \to\Div \tilde P = \Div P- \D_x(Q\cdot \D^{-1}_x \E(\mathscr{L})).
\end{equation}
Such a transformation amounts to adding a double zero to one of the components of $P$ so the resultant Lagrangian multiform will be essentially the same in that $\delta \textsf{dL}=0$ will give the same equations of motion, and $\textsf{dL}=0$ will still hold on these equations of motion.  This idea can be generalized further by noticing that the ``integration by parts'' can be carried out on any constituent part of $\tilde Q\cdot \E(\mathscr{L})$, e.g. 

\begin{equation}
\tilde Q_i\, \E_i(\mathscr{L}) \to -\D_x\tilde Q_i\, \D_x^{-1}\E_i(\mathscr{L}), 
\end{equation}
whilst leaving the resultant multiform essentially unchanged.  The $\tilde Q$ in \eqref{MFchar} is in evolutionary form with respect to $x_{p+1}$ i.e. it is in the form $u_{x_{p+1}}+Q(x,u^{(n)})=0$ where $Q(x,u^{(n)})$ does not contain $x_{p+1}$ or derivatives of $u$ with respect to $x_{p+1}$. If, by using the above operations we are able to put $\E(\mathscr{L})$ into evolutionary form with respect to some $x_j$, and neither $x_j$ nor derivatives of $u$ with respect to $x_j$ appear in $\tilde Q$ then we can reverse the roles of $\tilde Q$ and $\E(\mathscr{L})$ whilst essentially leaving the resultant multiform unchanged.  This idea forms the basis of the following theorem.

\begin{theorem}\label{thm2}
 Consider the Lagrangian and variational symmetry as given in Theorem \ref{mainthm} and let $j\in\{1,\dots,p\}$ be fixed. If there exist constants $a_k$ and multi-indices $J_k$ for $k=1,\ldots,q$ where the ${p+1}$ and $j$ components of each $J_k$ are zero, such that
 
\begin{equation}
a_k\D_{J_k}^{-1}\E_k(\mathscr{L})=0
\end{equation}
is in evolutionary form with respect to $x_j$, then the $q$ components of $\E(\mathscr{L}_{(\bar j)})$, up to re-ordering, are precisely the $q$ expressions

\begin{equation}
\frac{1}{a_k} \D_{J_k}\tilde Q_k.
\end{equation}
\end{theorem}

\begin{proof}
If there exist multi-indices $J_k$ and constants $a_k$ as described that put $\E(\mathscr{L})$ into evolutionary form with respect to $x_j$, then applying $a_k\D_{J_k}^{-1}$ to $\E_k(\mathscr{L})$ and $\frac{1}{a_k} \D_{J_k}$ to $\tilde Q_k$ in \eqref{MFchar} amounts to performing integration by parts on the products $\tilde Q_k\E_k(\mathscr{L})$, i.e.

\begin{equation}
\frac{1}{a_k} \D_{J_k}\tilde Q_k.a_k\D_{J_k}^{-1}\E_k(\mathscr{L})=\tilde Q_k\E_k(\mathscr{L})+\Div C_k
\end{equation}
for some $C_k$.  We note that the $j$ and ${p+1}$ components of $C_k$ are zero since the $j$ and ${p+1}$ components of each $J_k$ are zero.  It follows that 

\begin{equation}
\sum_{k=1}^q\frac{1}{a_k} \D_{J_k}\tilde Q_k.a_k\D_{J_k}^{-1}\E_k(\mathscr{L})=\Div \hat P
\end{equation}
where $\hat P=P+\sum_{k=1}^qC_k$.  Now that each $a_k\D_{J_k}^{-1}\E_k(\mathscr{L})$ is in evolutionary form, it follows from Noether's theorem that the corresponding characteristics represent variational symmetries of $\frac{1}{a_k} \D_{J_k}\tilde Q_k$, and by Theorem \ref{mainthm}, $\mathscr{L}_{(\bar j)}$ is the Lagrangian for $\frac{1}{a_k} \D_{J_k}\tilde Q_k$, $k=1,\ldots,q$.
\end{proof}

\noindent It follows that the multiforms described by $P$ and $\hat P$ in theorems \ref{mainthm} and \ref{thm2} both have $\mathscr{L}_{(\bar j)}$ and $\mathscr{L}$ as their $j$ and ${p+1}$ components respectively, since the $j$ and ${p+1}$ components of each $C_k$ are zero.

\subsection{The ``zero'' symmetry}

Every Lagrangian multiform we know of that has been considered up to this point has related to integrable system.  However, it is not the case that Lagrangian multiforms only exist for integrable systems, since Theorem \ref{mainthm} applies to any Lagrangian with a variational symmetry.  In fact, it turns out that every variational equation has at least one Lagrangian multiform description.\\

\noindent Using our construction, the requirements for a Lagrangian multiform are a Lagrangian density $\mathscr{L}(x,u^{(n)})$ and a variational symmetry \textbf{v}.  It is trivially true that the zero vector (i.e. $\textbf{v}_Q$ where $Q=0$) is a symmetry of every Lagrangian since $\textbf{v}_Q(\mathscr{L})=0$.  Letting $\tilde Q=u_{x_{p+1}}+Q=u_{x_{p+1}}$, it follows that

\begin{equation}\label{zerochar}
\tilde Q\cdot \E(\mathscr{L})=\Div P
\end{equation}
for some $P$, and it follows from Theorem \ref{mainthm} that $P$ describes a Lagrangian multiform.  Therefore every Lagrangian, regardless of integrability, fits into at least one Lagrangian multiform description.

\begin{remark}
This particular multiform could reasonably be described as semi-trivial, in that one of the equations of motion is simply $u_{x_{p+1}}=0$.  However, it does have a practical application relating to the inverse problem of finding a Lagrangian (if it exists) for a given equation of motion.  Also, the relation

\begin{equation}
\E(P\cdot Q)=\D_P^\ast(Q)+\D_Q^\ast(P),
\end{equation} 
as given in \cite{Olver1993} (where $\D_P(Q)$ is the Fr\'echet derivative of $P$ acting on $Q$ and $\D_P^\ast$ is the adjoint of $\D_P$) can be applied to \eqref{zerochar} in the case where $\tilde Q=u_{x_{p+1}}$ to derive the condition (also given in \cite{Olver1993}) that an equation has a Lagrangian description if and only if its Fr\'echet derivative is self adjoint.
\end{remark}
\noindent Since we can apply Theorem \ref{mainthm} with any variational symmetry, many Lagrangians can fit into more that one Lagrangian multiform description.  For example, if a given Lagrangian possesses time/space shift symmetries and rotational symmetries then we can obtain a Lagrangian multiform for each.  However, unless the symmetries themselves describe mutually commuting flows, we cannot expect it to be possible to connect these multiforms descriptions to each other in any coherent way (i.e. as we are able to do in the case of the AKNS multiform in section \ref{AKNSex}). The latter point emphasises the distinction between multiforms as just described, and multiforms carrying information about the integrability of the equations of motion, which was the original intent of the notion of Lagrangian multiforms.\\

\noindent Next, we shall give three examples of constructing Lagrangian multiforms from variational symmetries.  All three systems considered come from well known integrable hierarchies - this simplifies the task of finding variational symmetries, since the required symmetries are other equations taken from the respective hierarchies. 

\subsection{The sine-Gordon equation}
The sine-Gordon equation, $u_{x_1x_2}=\sin u$ with Lagrangian density

\begin{equation}
\mathscr{L}_{(12)}=\frac{1}{2}u_{x_1}u_{x_2}-\cos u
\end{equation}
and variational symmetry $Q=u_{3x_1}+\frac{1}{2}u_{x_1}^3$ is given as an example in \cite{Olver1993}.  We can confirm that $Q$ is a variational symmetry of $\mathscr{L}$ by checking that $\pr \textbf{v}_Q \mathscr{L}=\Div P$ for some $P$.  Indeed, we find that

\begin{equation}
\begin{split}
\pr \textbf{v}_Q \mathscr{L}=&\frac{1}{2}(u_{4x_1}+\frac{3}{2}u_{x_1}^2u_{x_1x_1})u_{x_2}+\frac{1}{2}(u_{3x_1x_2}+\frac{3}{2}u_{x_1}^2u_{x_1x_2})u_{x_1}+(u_{3x_1}+\frac{1}{2}u_{x_1}^3)\sin u\\
=&\D_{x_1}(\frac{1}{2}u_{x_1}u_{x_1x_1x_2}-\frac{1}{2}u_{x_1x_1}u_{x_1x_2}+\frac{1}{2}u_{x_1x_1x_1}u_{x_2}+\frac{1}{4}u_{x_1}^3u_{x_2}+u_{x_1x_1}\sin u -\frac{1}{2}u_{x_1}^2\cos u)\\
&+\D_{x_2}(\frac{1}{8}u_{x_1}^4).
\end{split}
\end{equation}
We now let $\tilde Q=u_{x_3}-Q$.  In this case, $\tilde Q=0$ is precisely the modified KdV equation which is known to be compatible with the sine-Gordon equation.  By Theorem \ref{mainthm}, the product

\begin{equation}\label{SGchar}
\tilde Q\cdot \E(\mathscr{L})=(u_{x_3}-u_{3x_1}-\frac{1}{2}u_{x_1}^3)(\sin u-u_{x_1x_2})=\Div P,
\end{equation}
i.e. it is a divergence.  If we write this product in terms of the components of $P$ we find that

\begin{equation}
P=
\begin{pmatrix}
-\frac{1}{2}u_{x_2}u_{x_3}+u_{x_1x_1}u_{x_1x_2}-u_{x_1x_1}\sin u+\frac{1}{2}u_{x_1}^2\cos u\\
-\frac{1}{2}u_{x_1}u_{x_3}-\frac{1}{2}u_{x_1x_1}^2+\frac{1}{8}u_{x_1}^4\\
\frac{1}{2}u_{x_1}u_{x_2}-\cos u
\end{pmatrix}
=
\begin{pmatrix}
\mathscr{L}_{(23)}\\
\mathscr{L}_{(31)}\\
\mathscr{L}_{(12)}
\end{pmatrix}
\end{equation}
satisfies \eqref{SGchar}, and is precicely the Lagrangian multiform for the sine-Gordon equation that was given in \cite{Suris2016a}.

\subsection{The AKNS multiform}\label{AKNSex}

The first two flows of the AKNS hierarchy \cite{AKNS1974} were shown to possess a Lagrangian multiform structure in \cite{Sleigh2019}.  The $\mathscr{L}_{(x_1x_2)}$ and $\mathscr{L}_{(x_3x_1)}$ AKNS Lagrangians, (see e.g. \cite{AVAN2016415}) are as follows:

\begin{equation}\label{AKL12}
\mathscr{L}_{(12)}= \frac{1}{2}(rq_{x_2}-qr_{x_2})+\frac{i}{2}q_{x_1}r_{x_1}+\frac{i}{2}q^2r^2\,,
\end{equation}
and
\begin{equation}\label{AKL31}
\mathscr{L}_{(31)}= \frac{1}{2}(qr_{x_3}-rq_{x_3})+\frac{1}{8}(r_{x_1}q_{x_1x_1}-q_{x_1}r_{x_1x_1})+\frac{3}{8}qr(rq_{x_1}-qr_{x_1})\,,
\end{equation}
giving equations of motion
\begin{eqnarray}
&&r_{x_2}=-\frac{i}{2}r_{x_1x_1}+ir^2q\,,\\
&&q_{x_2}=\frac{i}{2}q_{x_1x_1}-iq^2r
\end{eqnarray}
corresponding to the two components of $\E( \mathscr{L}_{(12)})=0$, and
\begin{eqnarray}
&&r_{x_3}=\frac{3}{2}rqr_{x_1}-\frac{1}{4}r_{x_1x_1x_1}\,,\\
&&q_{x_3}=\frac{3}{2}qrq_{x_1}-\frac{1}{4}q_{x_1x_1x_1}\,,
\end{eqnarray}
corresponding to the two components of $\E( \mathscr{L}_{(31)})=0$.  It is straightforward (but time consuming) to check that

\begin{equation}
\textbf{v}_Q=(\frac{3}{2}qrq_{x_1}-\frac{1}{4}q_{x_1x_1x_1})\frac{\partial}{\partial q} +(\frac{3}{2}rqr_{x_1}-\frac{1}{4}r_{x_1x_1x_1})\frac{\partial}{\partial r}
\end{equation}
is a variational symmetry of $\mathscr{L}_{(12)}$.  In order to apply Theorem \ref{mainthm} we define

\begin{equation}
\tilde Q=
\begin{pmatrix}
q_{x_3}\\r_{x_3}
\end{pmatrix}
-Q
\end{equation}
and it follows that

\begin{equation}\label{AKchar}
\begin{split}
\tilde Q\cdot E(\mathscr{L}_{(12)})=
\begin{pmatrix}
q_{x_3}-\frac{3}{2}qrq_{x_1}+\frac{1}{4}q_{x_1x_1x_1}\\
\\
r_{x_3}-\frac{3}{2}rqr_{x_1}+\frac{1}{4}r_{x_1x_1x_1}
\end{pmatrix}
\cdot 
\begin{pmatrix}
-r_{x_2}-\frac{i}{2}r_{x_1x_1}+ir^2q\\
\\
q_{x_2}-\frac{i}{2}q_{x_1x_1}+iq^2r
\end{pmatrix}
=\Div P
\end{split}
\end{equation}
for some P.  We find that

\begin{equation}
P=
\begin{pmatrix}
\mathscr{L}_{(23)}\\
\mathscr{L}_{(31)}\\
\mathscr{L}_{(12)}
\end{pmatrix}
\end{equation}
with

\begin{equation}
\begin{split}
\mathscr{L}_{(23)}=&\frac{1}{4}(q_{x_2}r_{x_1x_1}-r_{x_2}q_{x_1x_1})-\frac{i}{2}(q_{x_3}r_{x_1}+r_{x_3}q_{x_1})+\frac{1}{8}(q_{x_1}r_{x_1x_2}-r_{x_1}q_{x_1x_2})+\frac{3}{8}qr(qr_{x_2}-rq_{x_2})\\
&-\frac{i}{8}q_{x_1x_1}r_{x_1x_1}+\frac{i}{4}qr(qr_{x_1x_1}+rq_{x_1x_1})-\frac{i}{8}(q^2r_{x_1}^2+r^2q_{x_1}^2)+\frac{i}{4}qrq_{x_1}r_{x_1}-\frac{i}{2}q^3r^3.
\end{split}
\end{equation}
and $\mathscr{L}_{(12)}$ and $\mathscr{L}_{(31)}$ as given in \eqref{AKL12} and \eqref{AKL31} will satisfy \eqref{AKchar}.  This gives us the Lagrangian multiform

\begin{equation}
\textsf{L}=\mathscr{L}_{(12)}\textsf{ d} x_1\wedge \textsf{ d}x_2+\mathscr{L}_{(23)}\textsf{ d}x_2\wedge \textsf{ d}x_3+\mathscr{L}_{(31)}\textsf{ d}x_3\wedge \textsf{ d}x_1,
\end{equation}
for which $\textsf{dL}=0$ and $\delta\textsf{dL}=0$ as expected. This $3$-component multiform was first derived in \cite{Sleigh2019}.  We now follow a similar procedure to find the $\mathscr{L}_{(14)}$, $\mathscr{L}_{(24)}$ and $\mathscr{L}_{(34)}$ Lagrangians of the AKNS multiform, illustrating how our construction can be used to go beyond the first few terms in a Lagrangian multiform to include the higher flows of an integrable hierarchy. For the AKNS case, this means that we want to include the flow corresponding to the independent variable $x_4$ to produce the Lagrangian multiform
\begin{equation}\label{L1234}
\textsf{L}_{1234}=\mathscr{L}_{(12)}\textsf{ d} x_1\wedge \textsf{ d}x_2+
\mathscr{L}_{(13)}\textsf{ d} x_1\wedge \textsf{ d}x_3+
\mathscr{L}_{(14)}\textsf{ d} x_1\wedge \textsf{ d}x_4+
\mathscr{L}_{(23)}\textsf{ d} x_2\wedge \textsf{ d}x_3+
\mathscr{L}_{(24)}\textsf{ d} x_2\wedge \textsf{ d}x_4+
\mathscr{L}_{(34)}\textsf{ d} x_3\wedge \textsf{ d}x_4
\end{equation}
In order to find the $\mathscr{L}_{(14)}$, $\mathscr{L}_{(24)}$ and $\mathscr{L}_{(34)}$ we require our $\tilde Q$ to represent the $x_4$ flow of the hierarchy, i.e.

\begin{equation}
\tilde Q_4=
\begin{pmatrix}
q_{x_4}+i(\frac{3}{4}q^3r^2-\frac{1}{4}q^2r_{x_1x_1}-\frac{1}{2}qq_{x_1}r_{x_1}-qrq_{x_1x_1}-\frac{3}{4}rq_{x_1}^2+\frac{1}{8}q_{4x_1})\\
\\
r_{x_4}-i(\frac{3}{4}q^2r^3-\frac{1}{4}r^2q_{x_1x_1}-\frac{1}{2}rq_{x_1}r_{x_1}-qrr_{x_1x_1}-\frac{3}{4}qr_{x_1}^2+\frac{1}{8}r_{4x_1})
\end{pmatrix}.
\end{equation}
The components of $\tilde Q_4$ are obtained by using the recursive procedure given in \cite{FLASCHKA1983}.  Theorem \ref{mainthm} tells us that

\begin{equation}
\tilde Q_4\cdot \E (\mathscr{L}_{(12)})=\Div P^{124}
\end{equation}
where the components of $P^{124}$ (with respect to $x_1$, $x_2$ and $x_4$) are found to be 
\begin{subequations}
	\begin{equation}%\label{AKNS14}
	P^{124}_4= \frac{1}{2}(rq_{x_2}-qr_{x_2})+\frac{i}{2}q_{x_1}r_{x_1}+\frac{i}{2}q^2r^2,
	\end{equation}
	\begin{equation}\label{AKL41}
	\begin{split}
	P^{124}_2=& \frac{1}{2}(qr_{x_4}-rq_{x_4})+\frac{3i}{16}(q^2r_{x_1}^2+r^2q_{x_1}^2)+\frac{i}{4}qrq_{x_1}r_{x_1}+\frac{5i}{16}qr(qr_{x_1x_1}+rq_{x_1x_1})\\
	&-\frac{i}{8}q_{x_1x_1}r_{x_1x_1}-\frac{i}{4}q^3r^3
	\end{split}
	\end{equation}
	and
	\begin{equation}
	\begin{split}
	P^{124}_1=&\frac{3}{8}q^2r^2(rq_{x_1}-qr_{x_1})-\frac{i}{16}(q^2r_{x_1}r_{x_2}+r^2q_{x_1}q_{x_2})-\frac{5i}{16}qr(qr_{x_1x_2}+rq_{x_1x_2})\\
	&-\frac{1}{8}qr(rq_{3x_1}-qr_{3x_1})-\frac{1}{8}(q^2r_{x_1}r_{x_1x_1}-r^2q_{x_1}q_{x_1x_1})-\frac{1}{8}q_{x_1}r_{x_1}(rq_{x_1}-qr_{x_1})\\
	&\frac{1}{4}qr(r_{x_1}q_{x_1x_1}-q_{x_1}r_{x_1x_1})+\frac{3i}{8}qr(q_{x_1}r_{x_2}+r_{x_1}q_{x_2})-\frac{i}{8}(q_{3x_1}r_{x_2}+r_{3x_1}q_{x_2})\\
	&+\frac{1}{16}(q_{3x_1}r_{x_1x_1}-r_{3x_1}q_{x_1x_1})+\frac{i}{8}(q_{x_1x_1}r_{x_1x_2}+r_{x_1x_1}q_{x_1x_2})-\frac{i}{2}(q_{x_1}r_{x_4}+r_{x_1}q_{x_4}).
	\end{split}
	\end{equation}
\end{subequations}
We can now recognize $P^{124}_4=\mathscr{L}_{(12)}$ and we set $P^{124}_2=\mathscr{L}_{(41)}$ and $P^{124}_1=\mathscr{L}_{(24)}$, consistently with Theorem \ref{mainthm}.  From the construction of the coefficients, it follows immediately that for the multiform
\begin{equation}
\textsf{L}_{124}=\mathscr{L}_{(12)}\textsf{ d} x_1\wedge \textsf{ d}x_2+\mathscr{L}_{(24)}\textsf{ d}x_2\wedge \textsf{ d}x_4+\mathscr{L}_{(41)}\textsf{ d}x_4\wedge \textsf{ d}x_1,
\end{equation}
the multiform EL equations are satisfied when both $\E(\mathscr{L}_{(12)})=0$ and $\E(\mathscr{L}_{(41)})=0$, and that $\textsf{dL}_{124}=0$ on these equations of motion.\\

\noindent To produce the rest of the coefficients needed for $\textsf{L}_{1234}$, we now use the same $\tilde Q_4$ together with $\mathscr{L}_{(13)}$ to define $P^{134}$ such that
\begin{equation}
\tilde Q_4\cdot \E (\mathscr{L}_{(13)})=\Div P^{134}\,.
\end{equation}
Then we find that the components of $P^{134}$ (with respect to $x_1$, $x_3$ and $x_4$) are such that $P^{134}_4=\mathscr{L}_{(13)}=-\mathscr{L}_{(31)}$ given in \eqref{AKL31}, as expected from Theorem \ref{mainthm},   
\begin{equation}
\begin{split}
P^{134}_1\equiv\mathscr{L}_{(34)}=&\frac{i}{8}(q_{x_1x_1}r_{x_1x_3}+r_{x_1x_1}q_{x_1x_3})-\frac{i}{8}(q_{3x_1}r_{x_3}+r_{3x_1}q_{x_3})-\frac{i}{32}q_{3x_1}r_{3x_1}\\
&+\frac{i}{32}(q^2r_{x_1x_1}^2+r^2q_{x_1x_1}^2)+\frac{i}{32}q_{x_1}^2r_{x_1}^2
+\frac{3}{8}qr(rq_{x_4}-qr_{x_4})+\frac{9i}{32}q^4r^4\\
&-\frac{3i}{16}q^2r^2(qr_{x_1x_1}+rq_{x_1x_1})-\frac{i}{16}(q^2r_{x_1}r_{x_3}+r^2q_{x_1}q_{x_3})-\frac{5i}{16}qr(qr_{x_1x_3}+rq_{x_1x_3})\\
&+\frac{1}{4}(q_{x_1x_1}r_{x_4}-r_{x_1x_1}q_{x_4})+\frac{3i}{16}qr(q_{x_1}r_{3x_1}+r_{x_1}q_{3x_1})+\frac{i}{16}qrq_{x_1x_1}r_{x_1x_1}\\
&-\frac{i}{16}q_{x_1}r_{x_1}(qr_{x_1x_1}+rq_{x_1x_1})-\frac{15i}{16}q^2r^2q_{x_1}r_{x_1}+\frac{3i}{8}qr(q_{x_1}r_{x_3}+r_{x_1}q_{x_3})\\
&-\frac{1}{8}(q_{x_1}r_{x_1x_4}-r_{x_1}q_{x_1x_4})\,,
\end{split}
\end{equation}
and $P^{134}_3=\mathscr{L}_{(41)}$ {-} identical to the $\mathscr{L}_{(41)}$ previously identified as $P_2^{124}$, given in \eqref{AKL41}.  Again, from the construction of the coefficients, it follows immediately that for the multiform
\begin{equation}
\textsf{L}_{134}=\mathscr{L}_{(13)}\textsf{ d} x_1\wedge \textsf{ d}x_3+\mathscr{L}_{(34)}\textsf{ d}x_3\wedge \textsf{ d}x_4+\mathscr{L}_{(41)}\textsf{ d}x_4\wedge \textsf{ d}x_1,
\end{equation}
the multiform EL equations are satisfied when both $\E(\mathscr{L}_{(13)})=0$ and $\E(\mathscr{L}_{(41)})=0$, and also that $\textsf{dL}_{134}=0$ on these equations of motion.  We are now able to form the 6 component Lagrangian multiform $\textsf{L}_{1234}$ given in \eqref{L1234} and, as we would hope, the multiform EL equations are all consequences of $\E(\mathscr{L}_{(1i)})=0$ for $i\in\{2,3,4\}$, and $\textsf{dL}_{1234}=0$ on these equations.  Therefore, in this case, we were able to incorporate two commuting variational symmetries to extend our multiform, but will this always be possible?  Inspired by the AKNS example we have just carried out, we now examine this problem in the case where the $\mathscr{L}_{(12)}$ Lagrangian and variational symmetry characteristics are autonomous polynomials in the field variables and their derivatives.\\

\noindent Given that each $\textsf{L}_{1ij}$ is determined from $\textsf{dL}_{1ij}$, we have the freedom to add any exact 2-form to $\textsf{L}_{1ij}$ without affecting the multiform structure.  As a result, the $\mathscr{L}_{(1i)}, \mathscr{L}_{(ij)}$ and $\mathscr{L}_{(j1)}$ we obtain are not uniquely defined; this fact holds added significance when extending our multiform to include more than one commuting symmetry.  When forming $\textsf{L}_{123}$, any choice of $\mathscr{L}_{(12)}, \mathscr{L}_{(23)}$ and $\mathscr{L}_{(31)}$ such that $\textsf{dL}_{123}= \tilde Q\cdot E(\mathscr{L}_{(12)})\textsf{d}x_1\wedge\textsf{d}x_2\wedge\textsf{d}x_3$ will give us a valid multiform.  When we then form $\textsf{L}_{124}$, we now require that the $\mathscr{L}_{(12)}$ is exactly the same as the one in $\textsf{L}_{123}$.  This is not a problem, since we will always be able to make it so by adding an appropriate exact 2-form to $\textsf{L}_{124}$.  Similarly, when we come to form $\textsf{L}_{134}$, it will always be possible to get the same $\mathscr{L}_{(13)}$ that was obtained in $\textsf{L}_{123}$ by adding an appropriate exact 2-form.  However, it is not entirely obvious that the $\mathscr{L}_{(14)}$ obtained at this stage will be exactly the same as the one in $\textsf{L}_{124}$.  If the two $\mathscr{L}_{(14)}$ components were to differ by a total $x_4$ derivative then it would not be possible to correct this by adding an exact 2-form without also changing $\mathscr{L}_{(13)}$, which we don't want to do because it is already in the form we require.\\

\noindent In the case of a 2-form where $\mathscr{L}_{(12)}$ contains only $x_1$ and $x_2$ derivatives of $u$, it follows from the form of $\textsf{dL}_{12i}$, as given by Theorem \ref{mainthm}, that the resulting $\mathscr{L}_{(i1)}$ Lagrangian need only contain first order derivatives of $u$ with respect to $x_i$ and no products of $x_i$ derivatives of $u$.  This is because, when applying Theorem \ref{mainthm} to obtain $\textsf{dL}_{12i}$, the only $x_i$ derivatives of $u$ that appear come from

\begin{equation}
u_{x_i}\cdot \E(\mathscr{L}_{(12)}).
\end{equation}
When reversing the integration by parts that was used to obtain $\E(\mathscr{L}_{(12)})$ from $\mathscr{L}_{(12)}$, this becomes

\begin{equation}
\D_{x_i}\mathscr{L}_{(12)} +\D_{x_1}A_1+\D_{x_2}A_2
\end{equation}
for some $A_1$ and $A_2$, and since all integration by parts was with respect to $x_1$ and $x_2$, $A_1$ and $A_2$ do not contain $2^{nd}$ or higher order derivatives with respect to $x_i$, or products of $x_i$ derivatives of $u$.  This, in conjunction with the multiform EL equations, in particular those of the form

\begin{equation}
\frac{\delta \mathscr{L}_{(12)}}{\delta u_{x_2}}=\frac{\delta \mathscr{L}_{(1i)}}{\delta u_{x_i}}
\end{equation}
for $i> 1$, where
\begin{equation}
\frac{\delta \mathscr{L}_{(ij)}}{\delta u_I}=\sum_{q,r=0}^\infty(-1)^{q+r}\D_{x_i}^q\D_{x_j}^r\frac{\partial \mathscr{L}_{(ij)}}{\partial u_{Ii^qj^r}}
\end{equation}
tells us that, modulo total $x_1$ derivatives, all $\mathscr{L}_{(1i)}$ for $i>2$ are of the form

\begin{equation}
\frac{\delta \mathscr{L}_{(12)}}{\delta u_{x_2}}u_{x_i}+\mathscr{F}_i
\end{equation}
where $\mathscr{F}_i$ is some function that has no direct dependence on $x_i$ derivatives of $u$.  This guarantees that, for example, the $\mathscr{L}_{(14)}$ coming from $\textsf{L}_{134}$ can be made to coincide with the one coming from $\textsf{L}_{124}$.\\

\noindent There is also the question of whether the multiform EL equations and closure relation that relate to $\textsf{dL}_{234}$ will be satisfied on the equations of motion relating to $\mathscr{L}_{(12)}, \mathscr{L}_{(13)}$ and $\mathscr{L}_{(14)}$.  To show that this is the case, we follow a similar argument to the one given in \cite{Suris2016}.  Once all of the $\mathscr{L}_{(1i)}$'s are consistently defined, we can form $\textsf{L}_{1234}$ and it follows from 

\begin{equation}
\textsf{d}^2(\textsf{L}_{1234})=0
\end{equation}
 and the form of $\textsf{dL}_{123}, \textsf{dL}_{124}$ and $\textsf{dL}_{134}$ in terms of the $\mathscr{L}_{(ij)}$ that

\begin{equation}
\D_{x_1}(\D_{x_2}\mathscr{L}_{(34)}-\D_{x_3}\mathscr{L}_{(24)}+\D_{x_4}\mathscr{L}_{(23)})
\end{equation}
has a double zero on the equations of motion.  Then, since each $\mathscr{L}_{(ij)}$ is an autonomous polynomial, it follows that $\textsf{dL}_{234}$ also has a double zero on the equations of motion, so all of the required relations will be satisfied.  This argument can then be used iteratively to further extend the multiform to include higher flows relating to additional commuting variational symmetries.  It is also possible to extend this argument to the case of autonomous polynomial systems in higher dimensions, but it remains an open problem to extend this argument to non-autonomous, non-polynomial systems.

\subsection{The KP multiform}

In this section, we shall construct a Lagrangian multiform for the Kadomtsev-Petviashvili (KP) equation \cite{KP1970}. This is the first example of a Lagrangian multiform for an integrable PDE in $2+1$ dimensions. It is therefore a $3$-form. A Lagrangian multiform for the discretised KP equation is given in \cite{Quispel2009}.  Attempts to perform a continuum limit (see \cite{Vermeeren2019a} for examples of such a procedure) in order to obtain a continuous Lagrangian multiform for the KP equation have, so far, been unsuccessful.  In order to proceed, we take as our starting point the Lagrangians

\begin{subequations}
\begin{equation}\label{L123}
\mathscr{L}_{(123)}=\frac{1}{2}v_{x_1x_1}v_{x_1x_3}-\frac{1}{2}v_{3x_1}^2-\frac{1}{2}v_{x_1x_2}^2+v_{x_1x_1}^3
\end{equation}

\begin{equation}\label{L412}
\mathscr{L}_{(412)}=\frac{1}{2}v_{x_1x_1}v_{x_1x_4}-2v_{3x_1}v_{x_1x_1x_2}-\frac{2}{3}v_{x_1x_2}v_{x_2x_2}+4v_{x_1x_1}^2v_{x_1x_2}
\end{equation}
\end{subequations}
where $v_{3x_1}=v_{x_1x_1x_1}$.  These are based on the KP Hamiltonians given in \cite{Case1985}, which are based on the formulation of \cite{Lin1982}.  In order to avoid non-local terms, these Lagrangians are given in terms of $v$ such that $v_{x_1x_1}=q$, where $q$ is the usual KP field variable.  These Lagrangians give equations of motion 

\begin{subequations}
\begin{equation}\label{KP123}
v_{3x_1x_3}-v_{x_1x_1x_2x_2}+v_{6x_1}+6v_{3x_1}^2+6v_{x_1x_1}v_{4x_1}=0,
\end{equation}
the first KP equation, and
\begin{equation}\label{KP412}
v_{3x_1x_4}+4v_{5x_1x_2}-\frac{4}{3}v_{x_13x_2}+8v_{4x_1}v_{x_1x_2}+24v_{3x_1}v_{x_1x_1x_2}+16v_{x_1x_1}v_{3x_1x_2}=0
\end{equation}
\end{subequations}
the second KP equation respectively.  It is straightforward (although time consuming) to check that setting $Q$ equal to

\begin{equation}
D^{-3}_{x_1}(-v_{x_1x_1x_2x_2}+v_{6x_1}+6v_{3x_1}^2+6v_{2x_1}v_{4x_1})=-D^{-1}_{x_1}(v_{x_2x_2}+3v_{x_1x_1}^2)+v_{3x_1}
\end{equation}
gives a variational symmetry $\textbf{v}_Q$ of the second KP equation \eqref{KP412}.  This implies that

\begin{equation}
\begin{split}
&(v_{x_1x_1x_1x_4}+4v_{5x_1x_2}-\frac{4}{3}v_{x_13x_2}+8v_{4x_1}v_{x_1x_2}+24v_{3x_1}v_{x_1x_1x_2}+16v_{x_1x_1}v_{3x_1x_2})(v_{x_3}-D^{-1}_{x_1}(v_{x_2x_2}+3v_{x_1x_1}^2)+v_{3x_1})\\
&=\Div P
\end{split}
\end{equation}
We use integration by parts (i.e. integrate the first bracket and differentiate the second bracket, both with respect to $x_1$) to remove non-local terms and get

\begin{equation}\label{KPcharform}
\begin{split}
&(v_{x_1x_1x_4}+4v_{4x_1x_2}-\frac{4}{3}v_{3x_2}+8v_{3x_1}v_{x_1x_2}+16v_{x_1x_1}v_{x_1x_1x_2})(v_{x_1x_3}-v_{x_2x_2}+3v_{x_1x_1}^2+v_{4x_1})=\Div \tilde P
\end{split}
\end{equation}
As expected, $\tilde P$ describes a Lagrangian 3-form

\begin{equation}\label{KPmultiform}
\textrm{L}=\mathscr{L}_{(123)}\textrm{d}x_1\wedge \textrm{d}x_2\wedge \textrm{d}x_3+\mathscr{L}_{(234)}\textrm{d}x_2\wedge \textrm{d}x_3\wedge \textrm{d}x_4+\mathscr{L}_{(341)}\textrm{d}x_3\wedge \textrm{d}x_4\wedge \textrm{d}x_1+\mathscr{L}_{(412)}\textrm{d}x_4\wedge \textrm{d}x_1\wedge \textrm{d}x_2
\end{equation}
with the $1,2,3$ and $4$ components of $\tilde P$ corresponding to $-\mathscr{L}_{(234)},\mathscr{L}_{(341)},-\mathscr{L}_{(412)}$ and $\mathscr{L}_{(123)}$. The $\mathscr{L}_{(123)}$ and $\mathscr{L}_{(412)}$ Lagrangians are precisely those given in \eqref{L123} and \eqref{L412}.  We find that the $\mathscr{L}_{(234)}$ Lagrangian is given by

\begin{equation}
\begin{split}
\mathscr{L}_{(234)}=&-\frac{1}{2}v_{x_1x_3}v_{x_1x_4}-4v_{x_1x_3}v_{3x_1x_2}+2v_{x_1x_1x_3}v_{x_1x_1x_2}-\frac{2}{3}v_{x_2x_2}v_{x_2x_3}+v_{x_2x_2}v_{x_1x_4}\\
&+4v_{x_2x_2}v_{3x_1x_2}-\frac{8}{3}v_{x_1x_2x_2}v_{x_1x_1x_2}-v_{3x_1}v_{x_1x_1x_4}+\frac{4}{3}v_{3x_1}v_{3x_2}-4v_{3x_1}^2v_{x_1x_2}\\
&+8v_{x_1x_1}v_{3x_1}v_{x_1x_1x_2}+8v_{x_1x_1}v_{x_1x_2}v_{x_2x_2}+\frac{4}{3}v_{x_1x_2}^3-8v_{x_1x_1}v_{x_1x_2}v_{x_1x_3}-8v_{x_1x_1}^3v_{x_1x_2}
\end{split}
\end{equation}
and the $\mathscr{L}_{(341)}$ Lagrangian is given by

\begin{equation}
\begin{split}
\mathscr{L}_{(341)}=&\frac{2}{3}v_{x_2x_2}^2+2v_{4x_1}^2-2v_{3x_1}v_{x_1x_1x_3}-\frac{4}{3}v_{x_2x_2}v_{x_1x_3}-\frac{2}{3}v_{x_1x_2}v_{x_2x_3}+v_{x_1x_2}v_{x_1x_4}\\
&-\frac{4}{3}v_{x_1x_1x_2}^2+\frac{4}{3}v_{3x_1}v_{x_1x_2x_2}+12v_{x_1x_1}^2v_{4x_1}+4v_{3x_1}^2v_{x_1x_1}-4v_{x_1x_1}^2v_{x_2x_2}\\
&+4v_{x_1x_1}v_{x_1x_2}^2+4v_{x_1x_1}^2v_{x_1x_3}+10v_{x_1x_1}^4
\end{split}
\end{equation}
It is clear from \eqref{KPcharform} that $\textsf{dL}=0$ when either the first \eqref{KP123} or second \eqref{KP412} KP equation holds.  When both the first and second KP equations hold, the left hand side of \eqref{KPcharform} gives a double zero, so we also have that $\delta\textsf{dL}=0$.  As a consequence, all of the multiform EL equations hold.  This is the first ever example of a Lagrangian 3-form.\\

\noindent In theory it should be possible to produce an infinite Lagrangian multiform for the entire KP hierarchy.  However, it is expected that the increasing prevalence of non-local terms as one progresses up the hierarchy would result in non-local terms appearing in the multiform.  We were able to avoid such terms in this example by expressing our equations in terms of a ``double potential'' $v$ where $v_{x_1x_1}=q$, but it is expected that, even in terms of this $v$, non-local terms would appear in the Lagrangians for the equations of the higher flows of the hierarchy.  For any finite KP multiform, one can introduce a higher potential dependent variable (e.g. $w$ such that $w_{x_1x_1x_1} =q$) in order to avoid non-local terms.  However, it is fairly straightforward to extend the multiform EL equations to allow linear non-local terms, and this may be the best approach when considering the full KP hierarchy.

\section{Conclusion}
Given any Lagrangian and an associated variational symmetry, the method outlined in this paper allows us to construct a Lagrangian multiform.  As a consequence, we have shown that the existence of a Lagrangian multiform structure is not a sufficient condition for integrability.  However, by linking Lagrangian multiforms to variational symmetries, existing results relating symmetries to integrability can now be applied to Lagrangian multiforms of the type described in this paper.  Whilst we have shown that every variational symmetry leads to a Lagrangian multiform, the question of when the converse holds remains an open problem.  In this paper, we have only considered continuous systems; we anticipate that the Noether-type theorems that are known for discrete systems, such as those given in \cite{Dorodnitsyn2001}, may yield analogous results in for discrete Lagrangian multiforms.  Whilst finalising this paper, the paper \cite{Vermeeren2019} has appeared, which uses the ideas of Noether's theorem to give an algorithm for finding the extended Lagrangian 2-form structure (i.e. incorporating arbitrarily many flows) from an appropriate set of $\mathscr{L}_{(1j)}$ Lagrangians.

\appendix

\section{Lagrangian $k$-form EL equations}\label{MFeqns}

The multiform EL equations for a Lagrangian $k$-form were first published in \cite{Vermeeren2018}.  Here we present a new proof of those equations.  We let
\begin{equation}\label{kformL}
\textsf{L}=\sum_{1\leq l_1<\ldots <l_k\leq N}\mathscr{L}_{(l_1\ldots l_k)}\ \textsf{d}x_{l_1}\wedge\ldots\wedge \textsf{d}x_{l_k}.
\end{equation}
be a $k$-form on a manifold of $N$ independent coordinates $x_1,\ldots,x_N$ and dependent variable $u$.  Therefore

\begin{equation}
\textsf{dL}=\sum_{1\leq i_1<\ldots<i_{k+1}\leq N}A^{i_1\ldots i_{k+1}}\textsf{d}x_{i_1}\wedge\ldots\wedge\textsf{d}x_{i_{k+1}}
\end{equation}
where the $A^{i_1\ldots i_{k+1}}$ depend on the $\mathscr{L}_{(l_1\ldots l_k)}$ in the usual way, i.e.

\begin{equation}
A^{i_1\ldots i_{k+1}}=\sum_{\alpha=1}^{k+1}(-1)^{k(\alpha+1)}\D_{x_{i_\alpha}}\mathscr{L}_{(i_{\alpha+1}\ldots i_{k+1}i_1\ldots i_{\alpha-1})}.
\end{equation}
For a fixed $i_1,\ldots ,i_{k+1}$, we shall write $\mathscr{L}_{(\bar \alpha)}$ to denote $\mathscr{L}_{(i_{\alpha+1}\ldots i_{k+1}i_1\ldots i_{\alpha-1})}$.  We define the variational derivative with respect to $u_I$ \textit{acting on $\mathscr{L}_{(\bar \alpha)}$}

\begin{equation}\label{vardivdef}
\frac{\delta\mathscr{L}_{(\bar \alpha)}}{\delta u_{I}}=\sum_{\substack{  J\\j_{i_\alpha}=0}}(-\D)_J\frac{\partial\mathscr{L}_{(\bar \alpha)}}{\partial u_{IJ}},
\end{equation}
where $I$ is the usual $N$ component multi-index representing derivatives with respect to $x_1,\ldots,x_N$, and the multi-indices $J$ are such that components $j_i=0$ whenever $i\neq i_1,\ldots, i_{k+1}$, i.e. $J$ represents derivatives with respect to $x_{i_1}, \ldots,x_{i_{k+1}}$.  We define that $\dfrac{\delta\mathscr{L}_{(\bar i)}}{\delta u_{I}}=0$ in the case where any component of the multi-index $I$ is negative.  Note that by this definition, the variational derivative of $\mathscr{L}_{(i_{\alpha+1}\ldots i_{k+1}i_1\ldots i_{\alpha-1})}$ with respect to $u_I$ only sees derivatives of $u_I$ with respect to the variables $x_{i_{\alpha+1}},\ldots x_{i_{k+1}},x_{i_1}\ldots ,x_{ i_{\alpha-1}}$, even though derivatives with respect to other variables may appear in $\mathscr{L}_{(i_{\alpha+1}\ldots i_{k+1}i_1\ldots i_{\alpha-1})}$.

\begin{theorem}\label{mfeqnsthm}
The dependent variable $u$ is a critical point of the $k$-form $\textsf{L}$ as defined in \eqref{kformL} if and only if for all $i_1,\ldots i_{k+1}$ such that $1\leq i_1<\ldots<i_{k+1}\leq N$, and for all $I$,

\begin{equation}\label{MFELeqns}
\sum_{\alpha=1}^{k+1}(-1)^{\alpha k}\frac{\delta\mathscr{L}_{(\bar \alpha)}}{\delta u_{I\backslash i_\alpha}}=0
\end{equation}
\end{theorem}

\noindent In order to prove that these are the multiform EL equations, we will require the following lemma:

\begin{lemma}Let $1\leq i_1<\ldots<i_{k+1}\leq N$ be fixed.  For all multi-indices $I$,
\begin{equation}\label{lem1}
\frac{\partial \mathscr{L}_{(\bar \alpha)}}{\partial u_I}=\sum_{\substack{  J\\j_i\leq 1\\j_{i_\alpha}=0}}\D_J\frac{\delta \mathscr{L}_{(\bar \alpha)}}{\delta u_{IJ}}
\end{equation}
where the summation is over all multi-indices $J$ as defined for \eqref{vardivdef}, such that the $i_\alpha^{th}$ component of $J$ is zero and the non-zero $j_i$ are equal to 1.
\end{lemma}

\begin{proof}
We first notice that the partial derivative on the left hand side of \eqref{lem1} appears only once in the sum on the right hand side.  We now need to show that all other terms that appear on the right hand side of \eqref{lem1}, which are all of the form $\D_A\dfrac{\partial \mathscr{L}_{(\bar \alpha)}}{\partial u_{IA}}$ for some multi-index $A$, sum to zero.  To show this, we consider the term $\D_A\dfrac{\partial \mathscr{L}_{(\bar \alpha)}}{\partial u_{IA}}$, and let $r$ be the number of non-zero entries in $A$.  We notice that this term appears exactly once when $|J|=0$ with a factor of $(-1)^{|A|}$, exactly ${r\choose 1}$ times with a factor of $(-1)^{|A|+1}$ when $|J|=1$, exactly ${r\choose 2}$ times with a factor of $(-1)^{|A|+2}$ when $|J|=2$ etc... In total, this term appears with a factor of $\pm\sum_{i=0}^{r}(-1)^i{r\choose i}$.  It can easily be seen that this sum is zero by considering the binomial expansion of $(1-1)^r$.

\end{proof}

\begin{proof}{(of Theorem \ref{mfeqnsthm})}
For the first part of this proof, we will show that $\delta \textsf{dL}=0$ by following the argument given in \cite{Suris2016}.  We assume that \textsf{L} contains terms up to $n^{th}$ order derivatives of $u$, (i.e. \textsf{L} depends on $u_I$ with $|I|\leq n$).  Let $B$ be an arbitrary $k+1$ dimensional ball with surface $\partial B$.  We consider the action functional $S$ over the closed surface $\partial B$ such that

\begin{equation}
S[u]=\oint_{\partial B}\sf{L}
\end{equation}
We then apply Stokes' theorem to write $S$ in terms of an integral over B:

\begin{equation}
S[u]=\int_ B\sf{dL}
\end{equation}
and we look for solutions of

\begin{equation}\label{deltaS1}
\delta S=\int_{B}\delta\textsf{d}\textsf{L}=0
\end{equation}
Since this must hold for arbitrary variations (i.e. with no boundary constraints) for every ball $B$, it follows that $u$ is a critical point of \textsf{L} if and only if the integrand $\delta \sf{dL}=0$, where

\begin{equation}\label{deltadL}
\delta\textsf{dL}=\sum_{1\leq i_1<\ldots<i_{k+1}\leq N}\sum_I\frac{\partial A^{i_1\ldots i_{k+1}}}{\partial u_{I}}           \delta u_I\wedge\textsf{d}x_{i_1}\wedge\ldots\wedge\textsf{d}x_{i_{k+1}}.
\end{equation}
This is equivalent to the statement that for all $1\leq i_1<\ldots<i_{k+1}\leq N$, for all $I$,

\begin{equation}\label{deltadL}
\frac{\partial A^{i_1\ldots i_{k+1}}}{\partial u_{I}}=0 
\end{equation}

\noindent We could stop here, and use \eqref{deltadL} as our multiform EL equations.  Indeed, there are occasions where this is the most convenient formulation to use.  However, it is more illuminating to express this in terms of variational derivatives; by doing so we see more clearly the interplay between the constituent $\mathscr{L}_{(l_1\ldots l_k)}$ and see that a consequence of $\delta \sf{dL}=0$ is that $\E(\mathscr{L}_{(l_1\ldots l_k)})=0$ for each $\mathscr{L}_{(l_1\ldots l_k)}$.\\

\noindent For the second part of this proof, we show that, for any choice of $1\leq i_1<\ldots<i_{k+1}\leq N$, \eqref{deltadL} holds if and only if $\forall I$,
\begin{equation}\label{trip}
\sum_{\alpha=1}^{k+1}(-1)^{\alpha k}\frac{\delta\mathscr{L}_{(\bar \alpha)}}{\delta u_{I\backslash i_\alpha}}=0.
\end{equation}
To do this, we first show that \eqref{trip} holds for $|I|> n$.  We then use an inductive argument to show that if \eqref{trip} holds for $|I|> m$ then it also holds for $|I|=m$.  The converse (that \eqref{trip} $\implies$ \eqref{deltadL}) is then easily seen from the intermediary steps of the proof.\\

\noindent We begin by (arbitrarily) fixing $1\leq i_1<\ldots<i_{k+1}\leq N$ and noticing that for $|I|\geq n+2$, \eqref{trip} holds.  In fact all terms are zero since, by definition, there are no $n+1^{th}$ order derivatives in our multiform.  We now consider the relation $\dfrac{\partial A^{i_1\ldots i_{k+1}}}{\partial u_{I}}=0$ in the case where $|I|=n+1$.  In this case we find that

\begin{equation}
\frac{\partial A^{i_1\ldots i_{k+1}}}{\partial u_{I}}=\sum_{\alpha=1}^{k+1}(-1)^{\alpha k+1}\frac{\partial\mathscr{L}_{(\bar \alpha)}}{\partial u_{I\backslash i_\alpha}}
\end{equation}
since there are no $n+1^{th}$ order derivatives in the $\mathscr{L}_{(\bar \alpha)}$.  By setting this equal to zero, we see that \eqref{trip} holds in the case where $|I|=n+1$.\\

\noindent Our inductive hypothesis is that \eqref{trip} holds for $|I| > m$.  We now consider the relation $\dfrac{\partial A^{i_1\ldots i_{k+1}}}{\partial u_{I}}=0$ in the case where $|I| =m$.\\

\noindent We now notice that

\begin{equation}\label{steps1}
\begin{split}
\frac{\partial A^{i_1\ldots i_{k+1}}}{\partial u_{I}}&=\sum_{\alpha=1}^{k+1}(-1)^{\alpha k+1}\frac{\partial }{\partial u_{I}}\D_{x_{i_\alpha}}\mathscr{L}_{(\bar \alpha)}\\
&=\sum_{\alpha=1}^{k+1}(-1)^{\alpha k+1}\bigg\{\frac{\partial \mathscr{L}_{(\bar \alpha)}}{\partial u_{I\backslash i_\alpha}}+\D_{x_{i_\alpha}}\frac{\partial \mathscr{L}_{(\bar \alpha)}}{\partial u_I}\bigg\}\\
&=\sum_{\alpha=1}^{k+1}(-1)^{\alpha k+1}\bigg\{\frac{\partial \mathscr{L}_{(\bar \alpha)}}{\partial u_{I\backslash i_\alpha}}+\sum_{\substack{J\\j_i\leq 1\\j_{i_\alpha}=0}}\D_{Ji_\alpha}\frac{\delta \mathscr{L}_{(\bar \alpha)}}{\delta u_{IJ}}\bigg\}\\
&=\sum_{\alpha=1}^{k+1}(-1)^{\alpha k+1}\bigg\{\frac{\partial \mathscr{L}_{(\bar \alpha)}}{\partial u_{I\backslash i_\alpha}}+\sum_{\substack{J\\j_i\leq 1\\ j_{i_\alpha}=1}}\D_{J}\frac{\delta \mathscr{L}_{(\bar \alpha)}}{\delta u_{IJ\backslash i_\alpha}}\bigg\}\\
&=\sum_{\alpha=1}^{k+1}(-1)^{\alpha k+1}\bigg\{\frac{\partial \mathscr{L}_{(\bar \alpha)}}{\partial u_{I\backslash i_\alpha}}\bigg\}+\sum_{ \substack{J\\j_i\leq 1\\|J|>0}}~\sum_{\substack{\alpha\\j_{i_\alpha}>0}}(-1)^{\alpha k+1}D_{J}\frac{\delta \mathscr{L}_{(\bar \alpha)}}{\delta u_{IJ\backslash i_\alpha}}
\end{split}
\end{equation}
where we have made use of \eqref{lem1} in the third line, re-labeled $J$ in the fourth line and changed the order of the summation in the last.  We now apply the inductive hypothesis to get

\begin{equation}\label{steps2}
\begin{split}
\frac{\partial A^{i_1\ldots i_{k+1}}}{\partial u_{I}}&=\sum_{\alpha=1}^{k+1}(-1)^{\alpha k+1}\bigg\{\frac{\partial \mathscr{L}_{(\bar \alpha)}}{\partial u_{I\backslash i_\alpha}}\bigg\}+\sum_{\substack{J\\j_i\leq 1\\|J|>0}}~\sum_{\substack{\alpha\\j_{i_\alpha}=0}}(-1)^{\alpha k}\D_{J}\frac{\delta \mathscr{L}_{(\bar \alpha)}}{\delta u_{IJ\backslash i_\alpha}}\\
&=\sum_{\alpha=1}^{k+1}(-1)^{\alpha k+1}\bigg\{\frac{\partial \mathscr{L}_{(\bar \alpha)}}{\partial u_{I\backslash i_\alpha}}-\sum_{\substack{ J\\j_i\leq 1\\j_{i_\alpha}=0\\ |J| >0}}\D_{J}\frac{\delta \mathscr{L}_{(\bar \alpha)}}{\delta u_{IJ\backslash i_\alpha}}\bigg\}=0.
\end{split}
\end{equation}
Finally, we use \eqref{lem1} to express this as

\begin{equation}\label{split4}
\begin{split}
\frac{\partial A^{i_1\ldots i_{k+1}}}{\partial u_{I}}&=\sum_{\alpha=1}^{k+1}(-1)^{\alpha k+1}\frac{\delta \mathscr{L}_{(\bar \alpha)}}{\delta u_{I\backslash i_\alpha}}=0
\end{split}
\end{equation}
and we have shown that \eqref{trip} holds for $|I|=m$.  By induction, it follows that \eqref{trip} holds for all $I$.  The converse can easily be seen to hold by following the steps taken in \eqref{steps1}, \eqref{steps2} and \eqref{split4} in reverse order.\\

\noindent We have shown that the multiform EL equations \eqref{MFELeqns} for a given $1\leq i_1<\ldots<i_{k+1}\leq N$ are equivalent to $\delta A^{i_1\ldots i_{k+1}}=0$ for the same $1\leq i_1<\ldots<i_{k+1}\leq N$.  It follows that the multiform EL equations holding for all $1\leq i_1<\ldots<i_{k+1}\leq N$ is equivalent to $\delta \textsf{dL}=0$.
\end{proof}

\subsection*{Compliance with ethical standards}
On behalf of all authors, the corresponding author states that there is no conflict of interest.

\bibliographystyle{unsrt}

\bibliography{Noether1028}

\end{document}